\documentclass{article}
\usepackage{graphicx,verbatim, bussproofs,tikz,float, latexsym}
\usetikzlibrary{arrows}
\usepackage[left=3cm,top=3cm,right=3cm,bottom=2cm]{geometry}
\usepackage{pgfplots,multicol}
\usepackage{amsmath,amsthm}
\usepackage{amssymb}
\usepackage{stmaryrd}
\usepackage{proof}
\usepackage{cmll}
\usepackage{bussproofs}
\usepackage{authblk}
\DeclareSymbolFont{extraup}{U}{zavm}{m}{n}
\DeclareMathSymbol{\vardiamond}{\mathalpha}{extraup}{87}

\newcommand{\nomi}{\mathbf{i}}
\newcommand{\nomj}{\mathbf{j}}
\newcommand{\nomk}{\mathbf{k}}

\renewcommand{\phi}{\varphi}

\renewcommand{\epsilon}{\varepsilon}
\theoremstyle{definition}
\newtheorem{theorem}{Theorem}[section]
\newtheorem{lemma}[theorem]{Lemma}

\newtheorem{proposition}[theorem]{Proposition}

\newtheorem{example}[theorem]{Example}

\newtheorem{corollary}[theorem]{Corollary}

\newtheorem{definition}[theorem]{Definition}
\newtheorem{remark}[theorem]{Remark}

\title{Sahlqvist-Type Completeness Theory for Hybrid Logic with Binder}
\author{Zhiguang Zhao}

\affil{\small School of Mathematics and Statistics, Taishan University, Tai'an, 271000, China}
\affil{\small zhaozhiguang23@gmail.com}

\date{}

\begin{document}
\maketitle
\begin{abstract}
In the present paper, we continue the research in \cite{Zh21c} to develop the Sahlqvist-type completeness theory for hybrid logic with satisfaction operators and downarrow binders $\mathcal{L}(@, \downarrow)$. We define the class of skeletal Sahlqvist formulas for $\mathcal{L}(@, \downarrow)$ following the ideas in \cite{ConRob}, but we follow a different proof strategy which is purely proof-theoretic, namely showing that for every skeletal Sahlqvist formula $\phi$ and its hybrid pure correspondence $\pi$, $\mathbf{K}_{\mathcal{H}(@, \downarrow)}+\phi$ proves $\pi$, therefore $\mathbf{K}_{\mathcal{H}(@, \downarrow)}+\phi$ is complete with respect to the class of frames defined by $\pi$, using a restricted version of the algorithm $\mathsf{ALBA}^{\downarrow}$ defined in \cite{Zh21c}. 

\end{abstract}

\section{Introduction}

\paragraph{Hybrid logic} Hybrid logics \cite{BeBlWo06} have higher expressivity than modal logics where it is possible to talk about states in the model using \emph{nominals} that are true at exactly one state. There are also other connectives in hybrid logic which are used to increase the expressive power, e.g.\ the \emph{satisfaction operator} $@_{\mathbf{i}}\phi$ which intuitively reads ``at the world denoted by $\nomi$, $\phi$ is true'', and the \emph{downarrow binder} $\downarrow x.\phi$ which binds the current world and can refer to the world later in $\phi$. In the present paper, we use $\mathcal{L}$ to denote the language for hybrid logic with nominals, $\mathcal{L}(@)$ with nominals and satisfaction operators, $\mathcal{L}(@,\downarrow)$ with nominals, satisfaction operators and downarrow binders, and $\mathbf{K}_{\mathcal{H}}$, $\mathbf{K}_{\mathcal{H}(@)}$, $\mathbf{K}_{\mathcal{H}(@,\downarrow)}$ to denote their respective basic systems.

\paragraph{Correspondence theory} Correspondence theory started as a branch of the model theory of modal logic. We say that a modal formula $\phi$ corresponds to a first-order formula $\alpha$ if they are valid on exactly the same class of Kripke frames. Sahlqvist \cite{Sa75} and van Benthem \cite{vB83} gave a syntactic description of certain modal formulas (later called \emph{Sahlqvist formulas}) which have two nice properties: first of all, they have first-order correspondents, secondly, they axiomatize normal modal logics strongly complete with respect to the class of Kripke frames defined by them. 

\paragraph{Correspondence and completeness theory for hybrid logic} Existing literature on correspondence and completeness theory for hybrid logic is abundant, see \cite{BetCMaVi04,Co09,
CoGoVa06b,ConRob,GaGo93,GoVa01,Ho06,HoPa10,Ta05,tCMaVi06,Zh21c}. Gargov and Goranko proved that any extension of $\mathbf{K}_{\mathcal{H}}$ with pure axioms (formulas that contain nominals only but no propositional variables) is strongly complete. ten Cate and Blackburn \cite{tCBl06} showed that any pure extension of $\mathbf{K}_{\mathcal{H}(@)}$ and $\mathbf{K}_{\mathcal{H}(@,\downarrow)}$ are strongly complete. ten Cate, Marx and Viana \cite{tCMaVi06} proved that any extension of $\mathbf{K}_{\mathcal{H}(@)}$ with modal Sahlqvist formulas (with no nominals and with propositional variables only) is strongly complete, and that these two kinds of results cannot be combined in general, since there is a pure formula and a modal Sahlqvist formula which together axiomatize a Kripke-incomplete logic when added to $\mathbf{K}_{\mathcal{H}(@)}$. Conradie and Robinson \cite{ConRob} studied to what extent can these two results be combined in $\mathcal{L}(@)$, using algorithmic and algebraic method. Zhao \cite{Zh21c} studies the correspondence theory for $\mathcal{L}(@,\downarrow)$. 

\paragraph{Our contribution} The present paper continues the study in \cite{Zh21c} on the completeness theory in the spirit of \cite{ConRob}, using algorithmic method, which is based on the algorithm $\mathsf{ALBA}$ (Ackermann Lemma Based Algorithm) \cite{ConPalSou,CoGhPa14}, which computes the first-order correspondents of input formulas/inequalities and is guaranteed to succeed on Sahlqvist formulas/inequalities. However, our completeness proof follows a different strategy, which is not algebraic as in \cite{ConRob}, but purely proof-theoretic. We define the class of skeletal Sahlqvist formulas (which is a subclass of Sahlqvist formulas defined in \cite{Zh21c}) for $\mathcal{L}(@, \downarrow)$ following the ideas in \cite{ConRob}, show that for every skeletal Sahlqvist formula $\phi$ and its hybrid pure correspondence $\pi$, $\mathbf{K}_{\mathcal{H}(@, \downarrow)}+\phi$ proves $\pi$, therefore $\mathbf{K}_{\mathcal{H}(@, \downarrow)}+\phi$ is complete with respect to the class of frames defined by $\pi$, using a restricted version of the algorithm $\mathsf{ALBA}^{\downarrow}$ defined in \cite{Zh21c}. 

\paragraph{Structure of the paper} The structure of the paper is as follows: Section \ref{Sec:Prelim} presents preliminaries on hybrid logic with satisfaction operators and downarrow binders, including syntax, semantics and basic system $\mathbf{K}_{\mathcal{H}(@,\downarrow)}$. Section \ref{Sec:Prelim:ALBA} provides ingredients on algorithmic correspondence theory. Section \ref{Sec:Sahl} defines skeletal Sahlqvist inequalities. Section \ref{Sec:ALBA} gives the Ackermann Lemma Based Algorithm $\mathsf{ALBA}^{\downarrow}_{\mathsf{Restrict}}$ for $\mathcal{L}(@,\downarrow)$. Section \ref{Sec:Success} briefly sketch the proof that $\mathsf{ALBA}^{\downarrow}_{\mathsf{Restrict}}$ succeeds on skeletal Sahlqvist formulas. Section \ref{Sec:Completeness} proves that $\mathbf{K}_{\mathcal{H}(@, \downarrow)}$ extended with skeletal Sahlqvist formulas are strongly complete. Section \ref{Sec:Conclusion} gives conclusions.

\section{Preliminaries on hybrid logic with binder}\label{Sec:Prelim}

In the present section we collect the preliminaries on hybrid logic with binder. For more details, see \cite[Chapter 14]{BeBlWo06} and \cite{tC05}.

\subsection{Language and syntax}\label{Subsec:Lan:Syn}

\begin{definition}

Given three pariwise disjoint countably infinite sets $\mathsf{Prop}$ of propositional variables, $\mathsf{Svar}$ of state variables, $\mathsf{Nom}$ of nominals, the hybrid language $\mathcal{L}(@,\downarrow)$ is defined as follows:
$$\varphi::=p \mid x \mid \mathbf{i} \mid \bot \mid \top \mid \neg\varphi \mid \varphi\lor\varphi \mid \varphi\land\varphi \mid \varphi\to\varphi \mid \Diamond\varphi \mid \Box\varphi \mid @_{x}\varphi\mid @_{\mathbf{i}}\varphi \mid \downarrow x.\varphi,$$
where $p\in \mathsf{Prop}$, $x\in\mathsf{Svar}$, $\mathbf{i}\in\mathsf{Nom}$. 

We use $\vec p$ to denote a set of propositional variables and $\phi(\vec p)$ to indicate that the propositional variables that occur in $\phi$ are all in $\vec p$. We use $\mathsf{Prop}(\phi)$ to denote the set of all propositional variables occurring in $\phi$. We say that a formula is \emph{pure} if it contains no propositional variables. We define free and bound occurrences of state variables as usual, and say that a hybrid formula is a \emph{sentence} if it contains no free occurrences of state variables. We define $\sigma$ to be a \emph{substitution} that uniformly replaces propositional variables by formulas and terms (nominals or state variables) by terms. We use $\phi[\theta/p]$ to denote the substitution replacing $p$ by $\theta$ uniformly. We also use $\phi[\gamma/\delta]$ to denote the replacement of some occurrences of $\delta$ in $\phi$ by $\gamma$. In the present paper we will only consider the language with one unary modality.
\end{definition}

In the article, we will use \emph{inequalities} of the form $\phi\leq\psi$, where $\phi$ and $\psi$ are formulas, and \emph{quasi-inequalities} of the form $\phi_1\leq\psi_1\ \&\ \ldots\ \&\ \phi_n\leq\psi_n\ \Rightarrow\ \phi\leq\psi$. We will find it easy to work with inequalities $\phi\leq\psi$ in place of implicative formulas $\phi\to\psi$ in Section \ref{Sec:Sahl}.

\subsection{Semantics}\label{Subsec:Seman}

\begin{definition}
A \emph{Kripke frame} is a pair $\mathbb{F}=(W,R)$ where $W$ is a non-empty set called the \emph{domain} of $\mathbb{F}$, $R$ is a binary relation on $W$ called the \emph{accessibility relation}. A \emph{pointed Kripke frame} is a pair $(\mathbb{F}, w)$ where $w\in W$. A \emph{Kripke model} is a pair $\mathbb{M}=(\mathbb{F}, V)$ such that $V:\mathsf{Prop}\cup\mathsf{Nom}\to P(W)$ is a \emph{valuation} on $\mathbb{F}$ where for all nominals $\nomi\in\mathsf{Nom}$, $V(\nomi)\subseteq W$ is a singleton.

An assignment $g$ on $\mathbb{M}=(W,R,V)$ is a map $g:\mathsf{Svar}\to W$. For any assignment $g$,  any $x\in\mathsf{Svar}$, any $w\in W$, we define $g^{x}_{w}$ (the \emph{$x$-variant of $g$}) as follows: $g^{x}_{w}(x)=w$ and $g^{x}_{w}(y)=g(y)$ for all $y\in\mathsf{Svar}\setminus\{x\}$.

Now the satisfaction relation is given as follows: for any Kripke model $\mathbb{M}=(W,R,V)$, assignment $g$ on $\mathbb{M}$, $w\in W$, 

\begin{center}
\begin{tabular}{l c l}
$\mathbb{M},g,w\Vdash p$ & iff & $w\in V(p)$;\\
$\mathbb{M},g,w\Vdash x$ & iff & $g(x)=w$;\\
$\mathbb{M},g,w\Vdash\nomi$ & iff & $\{w\}=V(\nomi)$;\\
$\mathbb{M},g,w\Vdash \bot$ & : & never;\\
$\mathbb{M},g,w\Vdash \top$ & : & always;\\
$\mathbb{M},g,w\Vdash \neg\varphi$ & iff & $\mathbb{M},g,w\nVdash\varphi$;\\
$\mathbb{M},g,w\Vdash\varphi\lor\psi$ & iff & $\mathbb{M},g,w\Vdash \varphi$ or $\mathbb{M},g,w\Vdash\psi$;\\
$\mathbb{M},g,w\Vdash\varphi\land\psi$ & iff & $\mathbb{M},g,w\Vdash \varphi$ and $\mathbb{M},g,w\Vdash\psi$;\\
$\mathbb{M},g,w\Vdash\varphi\to\psi$ & iff & $\mathbb{M},g,w\nVdash \varphi$ or $\mathbb{M},g,w\Vdash\psi$;\\
$\mathbb{M},g,w\Vdash\Diamond\varphi$ & iff & $\exists v(Rwv\ \mbox{ and }\ \mathbb{M},g,v\Vdash\varphi)$;\\
$\mathbb{M},g,w\Vdash \Box\varphi$ & iff & $\forall v(Rwv\ \Rightarrow\ \mathbb{M},g,v\Vdash\varphi)$;\\
$\mathbb{M},g,w\Vdash @_{x}\varphi$ & iff & $\mathbb{M},g,g(x)\Vdash\varphi$;\\
$\mathbb{M},g,w\Vdash @_{\mathbf{i}}\varphi$ & iff & $\mathbb{M},g,V(\nomi)\Vdash\varphi$;\\
$\mathbb{M},g,w\Vdash \downarrow x.\varphi$ & iff & $\mathbb{M},g^{x}_{w},w\Vdash\varphi$.\label{page:downarrow}\\
\end{tabular}
\end{center}

For any formula $\phi$, we use $\llbracket\varphi\rrbracket^{\mathbb{M},g}=\{w\in W\mid \mathbb{M},g,w\Vdash\varphi\}$ to denote the \emph{truth set} of $\varphi$ in $(\mathbb{M},g)$. $\varphi$ is \emph{globally true} on $(\mathbb{M},g)$ (notation: $\mathbb{M},g\Vdash\varphi$) if $\mathbb{M},g,w\Vdash\varphi$ for every $w\in W$. $\varphi$ is \emph{valid} on a Kripke frame $\mathbb{F}$ (notation: $\mathbb{F}\Vdash\varphi$) if $\varphi$ is globally true on $(\mathbb{F},V,g)$ for each valuation $V$ and each assignment $g$.
\end{definition}

For the semantics of inequalities and quasi-inequalities, they are given as follows:
\begin{itemize}
\item $$\mathbb{M},g\Vdash\phi\leq\psi\mbox{ iff }$$$$(\mbox{for all }w\in W, \mbox{ if }\mathbb{M},g,w\Vdash\phi, \mbox{ then }\mathbb{M},g,w\Vdash\psi).$$
\item $$\mathbb{M},g\Vdash\phi_1\leq\psi_1\ \&\ \ldots\ \&\ \phi_n\leq\psi_n\ \Rightarrow\ \phi\leq\psi\mbox{ iff }$$
$$\mathbb{M},g\Vdash\phi\leq\psi\mbox{ holds whenever }\mathbb{M},g\Vdash\phi_i\leq\psi_i\mbox{ for all }1\leq i\leq n.$$
\end{itemize}

The definitions of validity are similar to formulas. It is easy to see that $\mathbb{M},g\Vdash\phi\leq\psi$ iff $\mathbb{M},g\Vdash\phi\to\psi$.

\subsection{Hilbert system}

The axioms and inference rules of the Hilbert system $\mathbf{K}_{\mathcal{H}(@,\downarrow)}$ of $\mathcal{L}(@,\downarrow)$ is given as follows (see \cite{tC05}):
\begin{itemize}
\item[(CT)] $\vdash\phi$ for all classical tautologies $\phi$
\item[(Dual)] $\vdash\Diamond p\leftrightarrow\neg\Box\neg p$
\item[(K)] $\vdash\Box(p\to q)\to(\Box p\to\Box q)$ 
\item[(K$_{@}$)] $\vdash @_{\nomi}(p\to q)\to(@_{\nomi} p\to@_{\nomi} q)$ 
\item[(Selfdual)] $\vdash\neg@_\nomi p\leftrightarrow@_{\nomi}\neg p$
\item[(Ref)] $\vdash@_{\nomi}\nomi$
\item[(Intro)] $\vdash\nomi\land p\to@_{\nomi}p$
\item[(Back)] $\vdash\Diamond@_{\nomi}p\rightarrow @_{\nomi}p$
\item[(Agree)] $\vdash@_\nomi@_\nomj p\to@_{\nomj}p$
\item[(DA)] $\vdash@_{\nomi}(\downarrow x.\phi\leftrightarrow\phi[\nomi/x])$
\item[(Name$_{\downarrow}$)] $\vdash\downarrow x.@_{x}\phi\to\phi$, if $x$ does not occur in $\phi$ 
\item[(BG$_{\downarrow}$)] $\vdash@_{\nomi}\Box\downarrow x.@_{\nomi}\Diamond x$
\item[(MP)] If $\vdash\phi\to\psi$ and $\vdash\phi$ then $\vdash\psi$
\item[(SB)] If $\vdash\phi$ then $\vdash\phi^{\sigma}$, provided that $\sigma$ is safe for $\phi$\footnote{We say that a substitution is \emph{safe} if it does not make free occurrences of state variables $x$ to be substituted into the scope of $\downarrow x$.}
\item[(Nec)] If $\vdash\phi$ then $\vdash\Box\phi$
\item[(Nec$_{@}$)] If $\vdash\phi$ then $\vdash @_{\nomi}\phi$
\item[(Nec$_{\downarrow}$)] If $\vdash\phi$ then $\vdash\downarrow x.\phi$
\end{itemize}

We use $\mathbf{K}_{\mathcal{H}(@,\downarrow)}+\Sigma$ to denote the logic system containing all axioms of $\mathbf{K}_{\mathcal{H}(@,\downarrow)}$ and $\Sigma$, closed under the rules of $\mathbf{K}_{\mathcal{H}(@,\downarrow)}$. We use $\vdash_{\Sigma}\phi$ to denote that $\phi$ is a theorem of $\mathbf{K}_{\mathcal{H}(@,\downarrow)}+\Sigma$. When $\Sigma$ is empty, we use the notation $\vdash\phi$.

We use $\Gamma\vdash_{\Sigma}\phi$ to denote that there are $\gamma_1,\ldots,\gamma_n\in\Gamma$ such that $\vdash_{\Sigma}\gamma_1\land\ldots\land\gamma_n\to\phi$. Given a frame class $\mathcal{F}$, we use $\Gamma\Vdash_{\mathcal{F}}\phi$ to denote that for any frame $\mathbb{F}=(W,R)\in\mathcal{F}$, any valuation $V$ and assignment $g$ on $\mathbb{F}$, any point $w\in W$, if $\mathbb{F},V,g,w\Vdash\gamma$ for all $\gamma\in\Gamma$, then $\mathbb{F},V,g,w\Vdash\phi$.

We have the following derived rules and theorems for any $\Sigma$, which will be useful in Section \ref{Sec:Completeness}:

\begin{itemize}
\item[(Trans)] $\vdash_{\Sigma}@_{\nomj}\phi\land@_{\nomi}\nomj\to@_{\nomi}\phi$
\item[(Sym)] $\vdash_{\Sigma}@_{\nomi}\nomj\to@_{\nomj}\nomi$
\item $\vdash_{\Sigma}@_\nomi(\beta\land\gamma)\leftrightarrow@_\nomi\beta\land@_\nomi\gamma$
\item $\vdash_{\Sigma}\neg@_\nomi(\alpha\lor\beta)\leftrightarrow\neg@_\nomi\alpha\land\neg@_\nomi\beta$
\item $\vdash_{\Sigma}@_{\nomj}\alpha\land@_{\nomi}\Diamond\nomj\to@_{\nomi}\Diamond\alpha$
\item $\vdash_{\Sigma}@_\nomi@_\nomj p\leftrightarrow @_{\nomj}p$
\item $\vdash_{\Sigma}@_{\nomi}\downarrow x.\phi\leftrightarrow @_{\nomi}\phi[\nomi/x]$
\item $\vdash_{\Sigma}@_{\nomj}\alpha\land\neg@_{\nomk}\beta\land\neg@_{\nomi}(\nomj\to\neg\nomk)\to\neg @_{\nomi}(\alpha\to\beta)$
\item $\vdash_{\Sigma}@_{\nomi_i}(\nomi_1\lor\ldots\lor\nomi_n)$ for $1\leq i\leq n$
\item $\vdash_{\Sigma}\neg@_{\nomi_i}(\neg\nomi_1\land\ldots\land\neg\nomi_n)$ for $1\leq i\leq n$
\item[(Res)] If $\vdash_{\Sigma}\gamma\leftrightarrow\delta$, then $\vdash_{\Sigma}\phi\leftrightarrow\phi[\gamma/\delta]$.
\end{itemize}

\begin{definition}[Soundness and Strong Completeness]
We say that $\mathbf{K}_{\mathcal{H}(@,\downarrow)}+\Sigma$ is \emph{sound} with respect to $\mathcal{F}$, if $\Gamma\vdash_{\Sigma}\phi$ implies that $\Gamma\Vdash_{\mathcal{F}}\phi$. We say that $\mathbf{K}_{\mathcal{H}(@,\downarrow)}+\Sigma$ is \emph{strongly complete} with respect to $\mathcal{F}$, if $\Gamma\Vdash_{\mathcal{F}}\phi$ implies that $\Gamma\vdash_{\Sigma}\phi$.
\end{definition}

\begin{theorem}[Theorem 9.4.4 in \cite{tC05}]\label{Completeness:Pure}
For any set $\Sigma$ of pure $\mathcal{L}(@,\downarrow)$-sentences, $\mathbf{K}_{\mathcal{H}(@,\downarrow)}+\Sigma$ is sound and strongly complete with respect to the class of frames defined by $\Sigma$.
\end{theorem}

\section{Ingredients of algorithmic correspondence}\label{Sec:Prelim:ALBA}

In this paper, we give a restricted version $\mathsf{ALBA}^{\downarrow}_{\mathsf{Restricted}}$ of the correspondence algorithm $\mathsf{ALBA}^{\downarrow}$ for hybrid logic with binder defined in \cite{Zh21c}. The algorithm $\mathsf{ALBA}^{\downarrow}_{\mathsf{Restricted}}$ transforms the input hybrid formula $\varphi\to\psi$ into an equivalent set of pure quasi-inequalities which does not contain occurrences of propositional variables\footnote{Notice that here we do not use the expanded modal language as in \cite{Zh21c}.}.

Since the purpose of the algorithm is to give Hilbert-style proof of the hybrid pure correspondence $\pi$ of the input skeletal Sahlqvist formula $\phi\to\psi$, the ingredients we will give is different from the one in \cite{Zh21c}. They can be listed as follows:

\begin{itemize}
\item An algorithm $\mathsf{ALBA}^{\downarrow}_{\mathsf{Restricted}}$ which transforms a given hybrid formula $\varphi\to\psi$ into equivalent pure quasi-inequalities $\mathsf{Pure}(\varphi\to\psi)$;
\item A syntactically identified class of inequalities on which the algorithm is successful;
\item A translation of the inequalities and quasi-inequalities involved in the algorithm into hybrid formulas;
\item A proof that for a given skeletal Sahlqvist formula $\phi\to\psi$, for each step of the execution of $\mathsf{ALBA}^{\downarrow}_{\mathsf{Restricted}}$, the translation of the resulting quasi-inequality is provable in $\mathbf{K}_{\mathcal{H}(@,\downarrow)}+(\phi\to\psi)$, therefore $\pi$ is provable in $\mathbf{K}_{\mathcal{H}(@,\downarrow)}+(\phi\to\psi)$.
\end{itemize}

In the remainder of the paper, we will give the definition of skeletal Sahlqvist inequalities (Section \ref{Sec:Sahl}), define a modified version of the algorithm $\mathsf{ALBA}^{\downarrow}_{\mathsf{Restricted}}$ (Section \ref{Sec:ALBA}), and success on Sahlqvist inequalities (Section \ref{Sec:Success}). We will give a translation of the inequalities and quasi-inequalities involved in the algorithm into hybrid formulas as well as prove that for a given skeletal Sahlqvist formula $\phi\to\psi$, for each step of the execution of $\mathsf{ALBA}^{\downarrow}_{\mathsf{Restricted}}$, the translation of the resulting quasi-inequality is provable in $\mathbf{K}_{\mathcal{H}(@,\downarrow)}+(\phi\to\psi)$ (Section \ref{Sec:Completeness}).

\section{Skeletal Sahlqvist inequalities}\label{Sec:Sahl}

In the present section, since we will use the algorithm $\mathsf{ALBA}^{\downarrow}_{\mathsf{Restricted}}$ which is based on the classsification of nodes in the signed generation trees of hybrid modal formulas, we will use the unified correspondence style definition (cf.\ \cite{CPZ:Trans,PaSoZh16,Zh21}) to define skeletal Sahlqvist formulas. We will collect all the necessary preliminaries on skeletal Sahlqvist formulas. For the sake of the algorithm, we will find it convenient to use inequalities $\phi\leq\psi$ instead of implicative formulas $\phi\to\psi$.

\begin{definition}[Order-type](cf.\ \cite[page 346]{CoPa12})
For any $n$-tuple of propositional variables $(p_1, \ldots, p_n)$, an order-type of $(p_1, \ldots, p_n)$ is an element $\epsilon$ in $\{1,\partial\}^{n}$. We call $p_i$ has order-type 1 with respect to $\epsilon$ if $\epsilon_i=1$, and write $\epsilon(i)=1$ or $\epsilon(p_i)=1$; we call $p_i$ has order-type $\partial$ with respect to $\epsilon$ if $\epsilon_i=\partial$, and write $\epsilon(i)=\partial$ or $\epsilon(p_i)=\partial$. We use $\epsilon^{\partial}$ to denote the opposite order-type of $\epsilon$ where $\epsilon^{\partial}(p_i)=1$ (resp.\ $\epsilon^{\partial}(p_i)=\partial$) iff $\epsilon(p_i)=\partial$ (resp.\ $\epsilon(p_i)=1$).
\end{definition}

\begin{definition}[Signed generation tree]\label{adef: signed gen tree}(cf.\ \cite[Definition 4]{CPZ:Trans})
The \emph{positive} (resp.\ \emph{negative}) {\em generation tree} of any given hybrid modal formula $\phi$ is defined as follows:

We first label the root of the generation tree of $\phi$ with $+$ (resp.\ $-$), then label the children nodes as below:
\begin{itemize}
\item If a node is labelled with $\lor, \land, \Box$, $\Diamond$, $\downarrow x$, then label the same sign to its children nodes;
\item If a node is labelled with $\neg$, then label the opposite sign to its child node;
\item If a node is labelled with $\to$, then label the opposite sign to the first child node and the same sign to the second child node;
\item If a node is labelled with $@$, then label the same sign to the second child node (notice that we do not label the first child node with nominal or state variable).
\end{itemize}
Nodes in signed generation trees are \emph{positive} (resp.\ \emph{negative}) if they are signed $+$ (resp.\ $-$).
\end{definition}

\begin{example}

The positive generation tree of $+\Diamond (p \lor \neg\Box q)\to\Diamond q$ is given in Figure \ref{fig:Church:Rosser}.

\begin{figure}[htb]
\centering
\begin{tikzpicture}
\tikzstyle{level 1}=[level distance=0.8cm, sibling distance=1.5cm]
\tikzstyle{level 2}=[level distance=0.8cm, sibling distance=1.5cm]
\tikzstyle{level 3}=[level distance=0.8cm, sibling distance=1.5cm]
 \node {$+\to$}         
              child{node{$-\Diamond$}
                     child{node{$-\lor$}
                           child{node{$-p$}}
                            child{node{$-\neg$}
                             child{node{$+\Box$}
                                child{node{$+q$}}}}}}
              child{node{$+\Diamond$}
                        child{node{$+q$}}} 
;
\end{tikzpicture}
\caption{Positive generation tree for $\Diamond (p \lor \neg\Box q)\to\Diamond q$}
\label{fig:Church:Rosser}
\end{figure}
\end{example}

We will use signed generation trees in the inequalities $\phi\leq\psi$, where we use the positive generation tree $+\phi$ and the negative generation tree $-\psi$. We call an inequality $\phi\leq\psi$ \emph{uniform} in a variable $p_i$ if all occurrences of $p_i$ in $+\phi$ and $-\psi$ have the same sign, and call $\phi\leq\psi$ $\epsilon$-\emph{uniform} in an array $\vec{p}$ if $\phi\leq\psi$ is uniform in $p_i$, occurring with the sign indicated by $\epsilon$ (i.e., $p_i$ has the sign $+$ (resp.\ $-$) if $\epsilon(p_i)=1$ (resp.\ $\partial$)), for each $p_i$ in $\vec{p}$.

For any order-type $\epsilon$ over $n$, any formula $\phi(p_1,\ldots p_n)$,  any $1 \leq i \leq n$, an \emph{$\epsilon$-critical node} in a signed generation tree $\ast\phi$ (where $\ast\in\{+,-\}$) is a leaf node $+p_i$ (when $\epsilon_i = 1$) or $-p_i$ (when $\epsilon_i = \partial$). An $\epsilon$-{\em critical branch} in a signed generation tree $\ast\phi$ is a branch from an $\epsilon$-critical node. The $\epsilon$-critical branches are those which the algorithm $\mathsf{ALBA}^{\downarrow}_{\mathsf{Restricted}}$ will solve for. We say that $\ast\phi$ {\em agrees with} $\epsilon$, and write $\epsilon(\ast\phi)$, if every leaf node with a propositional variable $p$ in the signed generation tree $\ast\phi$ is $\epsilon$-critical.

We use $+\psi\prec\ast\phi$ (resp.\ $-\psi\prec\ast\phi$) to denote that an occurrence of a subformula $\psi$ inherits the positive (resp.\ negative) sign from the signed generation tree $\ast\phi$. We use $\epsilon(\gamma) \prec \ast \phi$ (resp.\ $\epsilon^\partial(\gamma) \prec \ast \phi$) to denote that the signed generation subtree $\gamma$, with the sign inherited from $\ast \phi$, agrees with $\epsilon$ (resp.\ $\epsilon^\partial$). A propositional variable $p$ is \emph{positive} (resp.\ \emph{negative}) in $\phi$ if $+p\prec+\phi$ (resp.\ $-p\prec+\phi$) for all occurrences of $p$ in $\phi$.\label{page:epsilon:subtree}

\begin{definition}\label{adef:good:branches}(cf.\ \cite[Definition 5]{CPZ:Trans})
Nodes in signed generation trees are called \emph{skeletal nodes}, according to Table \ref{aJoin:and:Meet:Friendly:Table}. For the names of skeletal nodes, see \cite[Remark 3.24]{PaSoZh16}. A branch in a signed generation tree is called a \emph{skeletal branch} if it consists (apart from variable nodes) of skeletal nodes only.
\begin{table}
\begin{center}
\begin{tabular}{| c | c |}
\hline
Skeletal\\
\hline
\begin{tabular}{c c c c c c c c c c }
$+$ & $\vee$ & $\wedge$ &$\Diamond$ & $\neg$ & $\downarrow x$ & @\\
$-$ & $\wedge$ & $\vee$ &$\Box$ & $\neg$ & $\downarrow x$ & @ & $\to$\\
\end{tabular}\\
\hline
\end{tabular}
\end{center}
\caption{Skeletal nodes.}\label{aJoin:and:Meet:Friendly:Table}
\vspace{-1em}
\end{table}
\end{definition}

\begin{definition}[skeletal Sahlqvist inequalities\footnote{This name comes from \cite{ConRob}.} and formulas]\label{aInducive:Ineq:Def}(cf.\ \cite[Definition 2.4]{ConRob})
For any order-type $\epsilon$, the signed generation tree $*\phi(p_1,\ldots p_n)$ is \emph{$\epsilon$-skeletal Sahlqvist} if for all $1\leq i\leq n$, every $\epsilon$-critical branch with leaf $p_i$ is skeletal. An inequality $\phi\leq\psi$ is \emph{$\epsilon$-skeletal Sahlqvist} if the signed generation trees $+\phi$ and $-\psi$ are $\epsilon$-skeletal Sahlqvist. An inequality $\phi\leq\psi$ is \emph{skeletal Sahlqvist} if it is \emph{$\epsilon$}-skeletal Sahlqvist for some $\epsilon$. An implicative formula $\phi\to\psi$ is \emph{skeletal Sahlqvist} if $\phi\leq\psi$ is skeletal Sahlqvist.
\end{definition}

\begin{example}\label{Example:Sahlqvist}
Here we give an example of a skeletal Sahlqvist inequality for the order-type $\epsilon=(1,1)$, where the skeletal nodes are marked with $S$, and the leaf nodes of $\epsilon$-critical branches are marked with $C$. It is clear that the branch from $+p_2$ to $+\land$ and the branch from $+p_1$ to $+\land$ are both $\epsilon$-critical and skeletal.

\begin{figure}[htb]
\centering
\begin{multicols}{3}
\begin{tikzpicture}
\tikzstyle{level 1}=[level distance=0.8cm, sibling distance=2cm]
\tikzstyle{level 2}=[level distance=0.8cm, sibling distance=2cm]
\tikzstyle{level 3}=[level distance=0.8cm, sibling distance=1.5cm]
 \node {$+\land,S$} 
              child{node{$+\Diamond,S$}
                     child{node{$+p_1,C$}}}
              child{node{$+p_2,C$}}
;
\end{tikzpicture}

\columnbreak

$\leq$

\columnbreak

\begin{tikzpicture}
\tikzstyle{level 1}=[level distance=1cm, sibling distance=1cm]
\tikzstyle{level 2}=[level distance=1cm, sibling distance=1cm]
\tikzstyle{level 3}=[level distance=1cm, sibling distance=1cm]
 \node {$-\lor,S$}         
                     child{node{$-\Diamond$}
                           child{node{$-\Box$}
                                 child{node{$-\Diamond$}
                                       child{node{$-p_1$}}}}}
                     child{node{$-\Diamond$}
                           child{node{$-\Box$}
                                 child{node{$-\Diamond$}
                                       child{node{$-p_2$}}}}}
;
\end{tikzpicture}
\end{multicols}
\caption{(1,1)-skeletal Sahlqvist inequality $\Diamond p_1\land p_2\leq \Diamond\Box\Diamond p_1\lor\Diamond\Box\Diamond p_2$}
\label{fig:Sahlqvist}
\end{figure}
\end{example}

\section{The algorithm $\mathsf{ALBA}^{\downarrow}_{\mathsf{Restricted}}$}\label{Sec:ALBA}

In the present section, we define the modified version of the correspondence algorithm $\mathsf{ALBA}^{\downarrow}_{\mathsf{Restricted}}$ for hybrid logic with binder, which is a partial version of the algorithm $\mathsf{ALBA}^{\downarrow}$ in \cite{Zh21c}. First of all, the algorithm receives an input formula $\phi\to\psi$ and transforms it into an inequality $\phi\leq\psi$. Then the algorithm goes in three steps.

\begin{enumerate}

\item \textbf{Preprocessing and first approximation}:

In the generation tree of $+\phi$ and $-\psi$\footnote{The algorithm relies on signed generation trees in Section \ref{Sec:Sahl}. We will identify a signed formula with its signed generation tree.},

\begin{enumerate}
\item Apply the distribution rules:

\begin{enumerate}
\item Push down $+\Diamond, +\downarrow x, +@_{\nomi}, +@_{x}, -\neg, +\land, -\to$ by distributing them over nodes labelled with $+\lor$ which are skeletal nodes (see Figure \ref{Figure:distribution:rules}; notice that here we treat $@_{\nomi}$ and $@_{x}$ as unary modality with only the right branch as the input, and $\triangle\in\{\Diamond,\downarrow x, @_\nomi, @_x\}$), and

\item Push down $-\Box, -\downarrow x, -@_{\nomi}, -@_{x}, +\neg, -\lor, -\to$ by distributing them over nodes labelled with $-\land$ which are skeletal nodes (see Figure \ref{Figure:distribution:rules:2}; here $\triangle\in\{\Box,\downarrow x, @_\nomi, @_x\}$).

\end{enumerate}

\begin{figure}[htb]
\centering
\begin{multicols}{8}
\begin{tikzpicture}[scale=0.7]
\tikzstyle{level 1}=[level distance=1cm, sibling distance=1cm]
\tikzstyle{level 2}=[level distance=1cm, sibling distance=1cm]
\tikzstyle{level 3}=[level distance=1cm, sibling distance=1cm]
 \node {$+\triangle$}         
              child{node{$+\lor$}
                     child{node{$+\alpha$}}
                           child{node{$+\beta$}}}
;
\end{tikzpicture}
\columnbreak

$\Rightarrow$
\columnbreak

\begin{tikzpicture}[scale=0.7]
\tikzstyle{level 1}=[level distance=1cm, sibling distance=1cm]
\tikzstyle{level 2}=[level distance=1cm, sibling distance=1cm]
\tikzstyle{level 3}=[level distance=1cm, sibling distance=1cm]
 \node {$+\lor$}
              child{node{$+\triangle$}
                     child{node{$+\alpha$}}}
              child{node{$+\triangle$}
                           child{node{$+\beta$}}}
 ;
\end{tikzpicture}

\columnbreak

$\ $
\columnbreak

$\ $
\columnbreak

\begin{tikzpicture}[scale=0.7]
\tikzstyle{level 1}=[level distance=1cm, sibling distance=1cm]
\tikzstyle{level 2}=[level distance=1cm, sibling distance=1cm]
\tikzstyle{level 3}=[level distance=1cm, sibling distance=1cm]
 \node {$-\neg$}         
              child{node{$+\lor$}
                     child{node{$+\alpha$}}
                           child{node{$+\beta$}}}
 ;
\end{tikzpicture}

\columnbreak

$\Rightarrow$
\columnbreak

\begin{tikzpicture}[scale=0.7]
\tikzstyle{level 1}=[level distance=1cm, sibling distance=1cm]
\tikzstyle{level 2}=[level distance=1cm, sibling distance=1cm]
\tikzstyle{level 3}=[level distance=1cm, sibling distance=1cm]
 \node {$-\land$}
              child{node{$-\neg$}
                     child{node{$+\alpha$}}}
              child{node{$-\neg$}
                           child{node{$+\beta$}}}
 ;
\end{tikzpicture}
\end{multicols}

\centering
\begin{multicols}{8}
\begin{tikzpicture}[scale=0.7]
\tikzstyle{level 1}=[level distance=1cm, sibling distance=1cm]
\tikzstyle{level 2}=[level distance=1cm, sibling distance=1cm]
\tikzstyle{level 3}=[level distance=1cm, sibling distance=1cm]
 \node {$+\land$}         
              child{node{+$\alpha$}}
              child{node{$+\lor$}
                     child{node{$+\beta$}}
                           child{node{$+\gamma$}}}
 ;
\end{tikzpicture}

\columnbreak

$\ \ \ \ \ \ \ \ \Rightarrow$
\columnbreak

\begin{tikzpicture}[scale=0.7]
\tikzstyle{level 1}=[level distance=1cm, sibling distance=2cm]
\tikzstyle{level 2}=[level distance=1cm, sibling distance=1cm]
\tikzstyle{level 3}=[level distance=1cm, sibling distance=1cm]
 \node {$+\lor$}
              child{node{$+\land$}
                     child{node{$+\alpha$}}
                     child{node{$+\beta$}}}
              child{node{$+\land$}
                           child{node{$+\alpha$}}
                           child{node{$+\gamma$}}}
 ;
\end{tikzpicture}
\columnbreak

$\ $
\columnbreak

$\ $
\columnbreak

\begin{tikzpicture}[scale=0.7]
\tikzstyle{level 1}=[level distance=1cm, sibling distance=1cm]
\tikzstyle{level 2}=[level distance=1cm, sibling distance=1cm]
\tikzstyle{level 3}=[level distance=1cm, sibling distance=1cm]
 \node {$+\land$}         
              child{node{$+\lor$}
                     child{node{$+\alpha$}}
                           child{node{$+\beta$}}}
              child{node{+$\gamma$}}
;
\end{tikzpicture}
\columnbreak

$\ \ \ \ \ \ \ \ \Rightarrow$
\columnbreak

\begin{tikzpicture}[scale=0.7]
\tikzstyle{level 1}=[level distance=1cm, sibling distance=2cm]
\tikzstyle{level 2}=[level distance=1cm, sibling distance=1cm]
\tikzstyle{level 3}=[level distance=1cm, sibling distance=1cm]
 \node {$+\lor$}
              child{node{$+\land$}
                     child{node{$+\alpha$}}
                     child{node{$+\gamma$}}}
              child{node{$+\land$}
                           child{node{$+\beta$}}
                           child{node{$+\gamma$}}}
 ;
\end{tikzpicture}
\end{multicols}

\centering
\begin{multicols}{3}

\begin{tikzpicture}[scale=0.7]
\tikzstyle{level 1}=[level distance=1cm, sibling distance=1cm]
\tikzstyle{level 2}=[level distance=1cm, sibling distance=1cm]
\tikzstyle{level 3}=[level distance=1cm, sibling distance=1cm]
 \node {$-\to$}         
              child{node{$+\lor$}
                     child{node{$+\alpha$}}
                           child{node{$+\beta$}}}
              child{node{$-\gamma$}}
;
\end{tikzpicture}
\columnbreak

$\ \ \ \ \ \ \ \ \Rightarrow$
\columnbreak

\begin{tikzpicture}[scale=0.7]
\tikzstyle{level 1}=[level distance=1cm, sibling distance=2cm]
\tikzstyle{level 2}=[level distance=1cm, sibling distance=1cm]
\tikzstyle{level 3}=[level distance=1cm, sibling distance=1cm]
 \node {$-\land$}
              child{node{$-\to$}
                     child{node{$+\alpha$}}
                     child{node{$-\gamma$}}}
              child{node{$-\to$}
                           child{node{$+\beta$}}
                           child{node{$-\gamma$}}}
 ;
\end{tikzpicture}
\end{multicols}
\caption{Distribution rules for $+\lor$}
\label{Figure:distribution:rules}
\end{figure}

\begin{figure}[htb]

\centering
\begin{multicols}{8}
\begin{tikzpicture}[scale=0.7]
\tikzstyle{level 1}=[level distance=1cm, sibling distance=1cm]
\tikzstyle{level 2}=[level distance=1cm, sibling distance=1cm]
\tikzstyle{level 3}=[level distance=1cm, sibling distance=1cm]
 \node {$-\triangle$}         
              child{node{$-\land$}
                     child{node{$-\alpha$}}
                           child{node{$-\beta$}}}
 ;
\end{tikzpicture}

\columnbreak

$\ \ \ \ \ \ \ \ \Rightarrow$
\columnbreak

\begin{tikzpicture}[scale=0.7]
\tikzstyle{level 1}=[level distance=1cm, sibling distance=1cm]
\tikzstyle{level 2}=[level distance=1cm, sibling distance=1cm]
\tikzstyle{level 3}=[level distance=1cm, sibling distance=1cm]
 \node {$-\land$}
              child{node{$-\triangle$}
                     child{node{$-\alpha$}}}
              child{node{$-\triangle$}
                           child{node{$-\beta$}}}
 ;
\end{tikzpicture}

\columnbreak

$\ $
\columnbreak

$\ $
\columnbreak

\begin{tikzpicture}[scale=0.7]
\tikzstyle{level 1}=[level distance=1cm, sibling distance=1cm]
\tikzstyle{level 2}=[level distance=1cm, sibling distance=1cm]
\tikzstyle{level 3}=[level distance=1cm, sibling distance=1cm]
 \node {$+\neg$}         
              child{node{$-\land$}
                     child{node{$-\alpha$}}
                           child{node{$-\beta$}}}
;
\end{tikzpicture}
\columnbreak

$\Rightarrow$
\columnbreak

\begin{tikzpicture}[scale=0.7]
\tikzstyle{level 1}=[level distance=1cm, sibling distance=1cm]
\tikzstyle{level 2}=[level distance=1cm, sibling distance=1cm]
\tikzstyle{level 3}=[level distance=1cm, sibling distance=1cm]
 \node {$+\lor$}
              child{node{$+\neg$}
                     child{node{$-\alpha$}}}
              child{node{$+\neg$}
                           child{node{$-\beta$}}}
 ;
\end{tikzpicture}
\end{multicols}

\centering
\begin{multicols}{8}
\begin{tikzpicture}[scale=0.7]
\tikzstyle{level 1}=[level distance=1cm, sibling distance=1cm]
\tikzstyle{level 2}=[level distance=1cm, sibling distance=1cm]
\tikzstyle{level 3}=[level distance=1cm, sibling distance=1cm]
 \node {$-\lor$}         
              child{node{$-\alpha$}}
              child{node{$-\land$}
                     child{node{$-\beta$}}
                           child{node{$-\gamma$}}}
;
\end{tikzpicture}
\columnbreak

$\Rightarrow$
\columnbreak

\begin{tikzpicture}[scale=0.7]
\tikzstyle{level 1}=[level distance=1cm, sibling distance=2cm]
\tikzstyle{level 2}=[level distance=1cm, sibling distance=1cm]
\tikzstyle{level 3}=[level distance=1cm, sibling distance=1cm]
 \node {$-\land$}
              child{node{$-\lor$}
                     child{node{$-\alpha$}}
                     child{node{$-\beta$}}}
              child{node{$-\lor$}
                           child{node{$-\alpha$}}
                           child{node{$-\gamma$}}}
 ;
\end{tikzpicture}
\columnbreak

$\ $
\columnbreak

$\ $
\columnbreak

\begin{tikzpicture}[scale=0.7]
\tikzstyle{level 1}=[level distance=1cm, sibling distance=1cm]
\tikzstyle{level 2}=[level distance=1cm, sibling distance=1cm]
\tikzstyle{level 3}=[level distance=1cm, sibling distance=1cm]
 \node {$-\lor$}         
              child{node{$-\land$}
                     child{node{$-\alpha$}}
                           child{node{$-\beta$}}}
              child{node{$-\gamma$}}
 ;
\end{tikzpicture}

\columnbreak

$\Rightarrow$
\columnbreak

\begin{tikzpicture}[scale=0.7]
\tikzstyle{level 1}=[level distance=1cm, sibling distance=2cm]
\tikzstyle{level 2}=[level distance=1cm, sibling distance=1cm]
\tikzstyle{level 3}=[level distance=1cm, sibling distance=1cm]
 \node {$-\land$}
              child{node{$-\lor$}
                     child{node{$-\alpha$}}
                     child{node{$-\gamma$}}}
              child{node{$-\lor$}
                           child{node{$-\beta$}}
                           child{node{$-\gamma$}}}
 ;
\end{tikzpicture}
\end{multicols}

\centering
\begin{multicols}{3}
\begin{tikzpicture}[scale=0.7]
\tikzstyle{level 1}=[level distance=1cm, sibling distance=1cm]
\tikzstyle{level 2}=[level distance=1cm, sibling distance=1cm]
\tikzstyle{level 3}=[level distance=1cm, sibling distance=1cm]
 \node {$-\to$}         
              child{node{+$\alpha$}}
              child{node{$-\land$}
                     child{node{$-\beta$}}
                           child{node{$-\gamma$}}}
 ;
\end{tikzpicture}

\columnbreak

$\Rightarrow$
\columnbreak

\begin{tikzpicture}[scale=0.7]
\tikzstyle{level 1}=[level distance=1cm, sibling distance=2cm]
\tikzstyle{level 2}=[level distance=1cm, sibling distance=1cm]
\tikzstyle{level 3}=[level distance=1cm, sibling distance=1cm]
 \node {$-\land$}
              child{node{$-\to$}
                     child{node{$+\alpha$}}
                     child{node{$-\beta$}}}
              child{node{$-\to$}
                           child{node{$+\alpha$}}
                           child{node{$-\gamma$}}}
 ;
\end{tikzpicture}
\end{multicols}
\caption{Distribution rules for $-\land$}
\label{Figure:distribution:rules:2}
\end{figure}
\item Apply the splitting rules:
$$\infer{\alpha\leq\gamma\ \ \ \beta\leq\gamma}{\alpha\lor\beta\leq\gamma}
\qquad
\infer{\alpha\leq\beta\ \ \ \alpha\leq\gamma}{\alpha\leq\beta\land\gamma}
$$
\item Apply the monotone and antitone variable-elimination rules\footnote{Here the monotone and antitone variable elimination rules eliminate propositional variables $p$ where the inequality is semantically monotone or antitone with respect to $p$.}:
$$\infer{\alpha(\perp)\leq\beta(\perp)}{\alpha(p)\leq\beta(p)}
\qquad
\infer{\beta(\top)\leq\alpha(\top)}{\beta(p)\leq\alpha(p)}
$$
for $\beta(p)$ positive in $p$ and $\alpha(p)$ negative in $p$.
\end{enumerate}
We denote by $\mathsf{Preprocess}(\phi\leq\psi)$ the finite set $\{\phi_i\leq\psi_i\}_{i\in I}$ of inequalities obtained after applying the previous rules exhaustively. Then we apply the following \emph{first approximation rule} to every inequality in $\mathsf{Preprocess}(\phi\leq\psi)$:
$$\infer{\nomi_0\leq\phi_i\ \ \ \psi_i\leq \neg\nomi_1}{\phi_i\leq\psi_i}
$$
Here, $\nomi_0$ and $\nomi_1$ are fresh nominals. Now we get a set of sets of inequalities $\{\nomi_0\leq\phi_i, \psi_i\leq \neg\nomi_1\}_{i\in I}$. We call the set $\{\nomi_0\leq\phi_i, \psi_i\leq \neg\nomi_1\}$ \emph{system}.
\item \textbf{The reduction stage}:
In this stage, for each $\{\nomi_0\leq\phi_i, \psi_i\leq \neg\nomi_1\}$, we apply the following rules to prepare for eliminating all the proposition variables in $\{\nomi_0\leq\phi_i, \psi_i\leq\neg\nomi_1\}$:
\begin{enumerate}
\item \textbf{Substage 1: Decomposing the skeletal branch}
In the current substage, the following rules are applied to decompose the skeletal branches of the signed  skeletal Sahlqvist formula:
\begin{enumerate}
\item Splitting rules:
$$
\infer{\nomi\leq\beta\ \ \ \nomi\leq\gamma}{\nomi\leq\beta\land\gamma}
\qquad
\infer{\alpha\leq\neg\nomi\ \ \ \beta\leq\neg\nomi}{\alpha\lor\beta\leq\neg\nomi}
$$
\item Approximation rules:
$$
\infer{\nomj\leq\alpha\ \ \ \nomi\leq\Diamond\nomj}{\nomi\leq\Diamond\alpha}
\qquad
\infer{\nomj\leq\alpha\ \ \ x\leq\Diamond\nomj}{x\leq\Diamond\alpha}
\qquad
\infer{\alpha\leq\neg\nomj\ \ \ \Box\neg\nomj\leq\neg\nomi}{\Box\alpha\leq\neg\nomi}
\qquad
\infer{\alpha\leq\neg\nomj\ \ \ \Box\neg\nomj\leq\neg x}{\Box\alpha\leq\neg x}
$$
$$
\infer{\nomj\leq\alpha}{\nomi\leq @_{\nomj}\alpha}
\qquad
\infer{\nomj\leq\alpha}{x\leq @_{\nomj}\alpha}
\qquad
\infer{\alpha\leq\neg\nomj}{@_{\nomj}\alpha\leq\neg\nomi}
\qquad
\infer{\alpha\leq\neg\nomj}{@_{\nomj}\alpha\leq\neg x}
$$
$$
\infer{x\leq\alpha}{\nomi\leq @_{x}\alpha}
\qquad
\infer{x\leq\alpha}{y\leq @_{x}\alpha}
\qquad
\infer{\alpha\leq\neg x}{@_{x}\alpha\leq\neg\nomi}
\qquad
\infer{\alpha\leq\neg x}{@_{x}\alpha\leq\neg y}
$$
$$
\infer{\nomi\leq\alpha[\nomi/x]}{\nomi\leq \downarrow x.\alpha}
\qquad
\infer{y\leq\alpha[y/x]}{y\leq \downarrow x.\alpha}
\qquad
\infer{\alpha[\nomi/x]\leq\neg\nomi}{\downarrow x.\alpha\leq\neg\nomi}
\qquad
\infer{\alpha[y/x]\leq\neg y}{\downarrow x.\alpha\leq\neg y}
$$
$$
\infer{\nomj\leq\alpha\ \ \ \ \ \ \ \beta\leq\neg\nomk\ \ \ \ \ \ \ \nomj\rightarrow\neg\nomk\leq\neg\nomi}{\alpha\rightarrow\beta\leq\neg\nomi}
$$
$$\infer{\nomj\leq\alpha\ \ \ \ \ \ \ \beta\leq\neg\nomk\ \ \ \ \ \ \ \nomj\rightarrow\neg\nomk\leq\neg x}{\alpha\rightarrow\beta\leq\neg x}
$$
The nominals introduced by the approximation rules must not occur in the system before applying the rule, and $\alpha[\nomi/x]$ (resp.\ $\alpha[y/x]$) indicates that all occurrences of $x$ in $\alpha$ are replaced by $\nomi$ (resp.\ $y$).
\item Residuation rules:
$$
\infer{\alpha\leq\neg\nomi}{\nomi\leq\neg\alpha}
\qquad
\infer{\nomi\leq\alpha}{\neg\alpha\leq\neg\nomi}
\qquad
\infer{\alpha\leq\neg x}{x\leq\neg\alpha}
\qquad
\infer{x\leq\alpha}{\neg\alpha\leq\neg x}
$$
\end{enumerate}
\item \textbf{Substage 2: The Ackermann stage}\footnote{In the Ackermann stage, for the sake of simplicity, we use $\nomi$ to denote both nominals and state variables, since their behaviours at this stage are essentially the same.}

In the present substage, we compute the minimal/maximal valuations for propositional variables and use the Ackermann rules to eliminate all the propositional variables. The two rules are the core of $\mathsf{ALBA}^{\downarrow}_{\mathsf{Restricted}}$, since their applications eliminate propositional variables. In fact, the previous substage aims at reaching a shape where the Ackermann rules can be applied. Notice that the Ackermann rules are executed on the whole set of inequalities rather than on a single inequality.\\

The right-handed Ackermann rule:

The system 
$\left\{ \begin{array}{ll}
\nomi_1\leq p \\
\vdots\\
\nomi_n\leq p \\
\nomj_1\leq\gamma_1\\
\vdots\\
\nomj_m\leq\gamma_m\\
\beta_1\leq\neg\nomk_1\\
\vdots\\
\beta_k\leq\neg\nomk_k\\
\end{array} \right.$ 
is replaced by 
$\left\{ \begin{array}{ll}
\nomj_1\leq\gamma_1[(\nomi_1\lor\ldots\lor\nomi_n)/p]\\
\vdots\\
\nomj_m\leq\gamma_m[(\nomi_1\lor\ldots\lor\nomi_n)/p]\\
\beta_1[(\nomi_1\lor\ldots\lor\nomi_n)/p]\leq\neg\nomk_1\\
\vdots\\
\beta_k[(\nomi_1\lor\ldots\lor\nomi_n)/p]\leq\neg\nomk_k\\

\end{array} \right.$

where each $\beta_i$ is positive, and each $\gamma_j$ negative in $p$;\\

The left-handed Ackermann rule:

The system
$\left\{ \begin{array}{ll}
p\leq\neg\nomi_1 \\
\vdots\\
p\leq\neg\nomi_n \\
\nomj_1\leq\gamma_1\\
\vdots\\
\nomj_m\leq\gamma_m\\
\beta_1\leq\neg\nomk_1\\
\vdots\\
\beta_k\leq\neg\nomk_k\\
\end{array} \right.$
is replaced by
$\left\{ \begin{array}{ll}
\nomj_1\leq\gamma_1[(\neg\nomi_1\land\ldots\land\neg\nomi_n)/p]\\
\vdots\\
\nomj_m\leq\gamma_m[(\neg\nomi_1\land\ldots\land\neg\nomi_n)/p]\\
\beta_1[(\neg\nomi_1\land\ldots\land\neg\nomi_n)/p]\leq\neg\nomk_1\\
\vdots\\
\beta_m[(\neg\nomi_1\land\ldots\land\neg\nomi_n)/p]\leq\neg\nomk_k\\
\end{array} \right.$

where each $\beta_i$ is negative, and each $\gamma_j$ positive in $p$.
\end{enumerate}
\item \textbf{Output}: If in the previous stage, for some $\{\nomi_0\leq\phi_i, \psi_i\leq \neg\nomi_1\}$, the algorithm gets stuck, i.e.\ some propositional variables cannot be eliminated by the  reduction rules, then the algorithm stops and output ``failure''. Otherwise, each initial tuple $\{\nomi_0\leq\phi_i, \psi_i\leq \neg\nomi_1\}$ of inequalities after the first approximation has been reduced to a set of pure inequalities $\mathsf{Reduce}(\phi_i\leq\psi_i)$, and then the output is a set of pure quasi-inequalities $\{\&\mathsf{Reduce}(\phi_i\leq\psi_i)\Rightarrow \nomi_0\leq \neg\nomi_1: \phi_i\leq\psi_i\in\mathsf{Preprocess}(\phi\leq\psi)\}$. Finally we uniformly substitute all free occurrences of state variables by fresh nominals, and denote the set of pure quasi-inequalities $\mathsf{Pure}(\phi\to\psi)$.
\end{enumerate}

Since the algorithm $\mathsf{ALBA}^{\downarrow}_{\mathsf{Restricted}}$ is a restricted version of the algorithm $\mathsf{ALBA}^{\downarrow}$, its soundness follows from the soundness of $\mathsf{ALBA}^{\downarrow}$.

\begin{theorem}[Soundness of the algorithm]\label{Thm:Soundness}
If $\mathsf{ALBA}^{\downarrow}_{\mathsf{Restricted}}$ runs successfully on $\phi\to\psi$ and outputs $\mathsf{Pure}(\phi\leq\psi)$, then for any Kripke frame $\mathbb{F}=(W,R)$, $$\mathbb{F}\Vdash\phi\to\psi\mbox{ iff }\mathbb{F}\models\mathsf{Pure}(\phi\to\psi).$$
\end{theorem}

\begin{remark}
The special feature of the restricted version of the algorithm $\mathsf{ALBA}^{\downarrow}_{\mathsf{Restricted}}$ compared with $\mathsf{ALBA}^{\downarrow}$ in \cite{Zh21c} is that there is no expanded hybrid language needed in $\mathsf{ALBA}^{\downarrow}_{\mathsf{Restricted}}$, and there is no tense operators needed due to the fact that we do not need most of the residuation rules in $\mathsf{ALBA}^{\downarrow}$ except for $\neg$. Another feature of $\mathsf{ALBA}^{\downarrow}_{\mathsf{Restricted}}$ is that during Stage 2, for each inequality involved, they are of the form $\nomi\leq\gamma$, $x\leq\gamma$, $\gamma\leq\neg\nomi$ or $\gamma\leq\neg x$, which means that they can be equivalently translated into hybrid formulas of the form $@_{\nomi}\gamma$, $@_x\gamma$, $\neg@_\nomi \gamma$ or $\neg@_x \gamma$, as we can see in Section \ref{Sec:Success} and \ref{Sec:Completeness}.
\end{remark}

\section{Success of $\mathsf{ALBA}^{\downarrow}_{\mathsf{Restricted}}$}\label{Sec:Success}

In the present section we show that $\mathsf{ALBA}^{\downarrow}_{\mathsf{Restricted}}$ succeeds on all skeletal Sahlqvist inequalities. The proof is similar to \cite[Section 7]{Zh21c}, but we will stress on the special shape of the inequalities involved in the execution of the algorithm.

\begin{theorem}\label{Thm:Success}
$\mathsf{ALBA}^{\downarrow}_{\mathsf{Restricted}}$ succeeds on all skeletal Sahlqvist formulas.
\end{theorem}

\begin{definition}[Definite $\epsilon$-skeletal Sahlqvist inequality]
Given an order-type $\epsilon$ and $\ast\in\{-,+\}$, the signed generation tree $\ast\phi$ is \emph{definite $\epsilon$-skeletal Sahlqvist} if it is $\epsilon$-skeletal Sahlqvist and there is no $+\lor,-\land$ occurring on an $\epsilon$-critical branch. An inequality $\phi\leq\psi$ is \emph{definite $\epsilon$-skeletal Sahlqvist} if $+\phi$ and $-\psi$ are both definite $\epsilon$-skeletal Sahlqvist.
\end{definition}

\begin{lemma}\label{Lemma:Stage:1}
Let $\{\phi_i\leq\psi_i\}_{i\in I}=\mathsf{Preprocess}(\phi\leq\psi)$ obtained by exhaustive application of the rules in Stage 1 on an input $\epsilon$-skeletal Sahlqvist inequality $\phi\leq\psi$. Then each $\phi_i\leq\psi_i$ is a definite $\epsilon$-skeletal Sahlqvist inequality.
\end{lemma}

\begin{proof}
The proof is essentially the same as in \cite[Lemma 7.3]{Zh21c}. 
\end{proof}

\begin{lemma}\label{Lemma:Substage:2:1}
Given inequalities $\nomi_0\leq\phi_i$ and $\psi_i\leq\neg\nomi_1$ obtained from Stage 1 where $+\phi_i$ and $-\psi_i$ are definite $\epsilon$-skeletal Sahlqvist, by applying the rules in Substage 1 of Stage 2 exhaustively, the inequalities obtained are in one of the following forms:

\begin{enumerate}
\item pure inequalities of the form $\nomi\leq\gamma$, $x\leq\gamma$, $\gamma\leq\neg\nomi$ or $\gamma\leq\neg x$, where $\gamma$ is pure;
\item inequalities of the form $\nomi\leq p$ or $x\leq p$ where $\epsilon(p)=1$;
\item inequalities of the form $p\leq\neg\nomi$ or $p\leq\neg x$ where $\epsilon(p)=\partial$;
\item inequalities of the form $\nomi\leq\delta$ or $x\leq\delta$ where $+\delta$ is $\epsilon^{\partial}$-uniform;
\item inequalities of the form $\delta\leq\neg\nomi$ or $\delta\leq\neg x$ where $-\delta$ is $\epsilon^{\partial}$-uniform.
\end{enumerate}
\end{lemma}

\begin{proof}
The proof is similar to \cite[Lemma 7.5]{Zh21c}. The rules in the Substage 1 of Stage 2 treat skeletal nodes in $+\phi_i$ and $-\psi_i$ except $+\lor$, $-\land$. For each rule, without loss of generality, we suppose that we start with an inequality of the form $\nomi\leq\alpha$. By applying the rules in Substage 1 of Stage 2, the inequalities we obtain are either a pure inequality (i.e.\ without propositional variables), or an inequality in which the left-hand side (resp.\ right-hand side) is $\nomi$ or $x$ (resp.\ $\neg\nomi$ or $\neg x$), and the other side of the inequality is a formula $\alpha'$ that is a subformula of $\alpha$, such that $\alpha'$ has one root connective less than $\alpha$. In addition, if $\alpha'$ is on the left-hand side (resp.\ right-hand side) then $-\alpha'$ ($+\alpha'$) is definite $\epsilon$-skeletal Sahlqvist.

By exhaustively applying the rules in the Substage 1 of Stage 2, we eliminate all the skeletal connectives in the $\epsilon$-critical branches, so for non-pure inequalities, they become of form 2, 3, 4 or 5.

In addition, in each inequality, either the left-hand side is $\nomi$ or $x$, or the right-hand side is $\neg\nomi$ or $\neg x$, and for each step of Substage 1 of Stage 2, after the applications of the rules, the resulting inequalities still have this property. Therefore, the final pure inequalities are of the form $\nomi\leq\gamma$, $x\leq\gamma$, $\gamma\leq\neg\nomi$ or $\gamma\leq\neg x$, where $\gamma$ is pure.
\end{proof}

\begin{lemma}\label{Lemma:Substage:2:4}
Suppose we have inequalities of the form in Lemma \ref{Lemma:Substage:2:1}, then the Ackermann lemmas are applicable and all propositional variables can be eliminated, and for each inequality in the system, either the left-hand side is $\nomi$ or $x$, or the right-hand side is $\neg\nomi$ or $\neg x$.
\end{lemma}

\begin{proof}
Easy observation from the syntactic requirements of the Ackermann lemmas.
\end{proof}

\begin{proof}[Proof of Theorem \ref{Thm:Success}]
Assume we have an $\epsilon$-skeletal Sahlqvist formula $\phi\to\psi$ as input. By Lemma \ref{Lemma:Stage:1}, we get a set of definite $\epsilon$-skeletal Sahlqvist inequalities. Then by Lemma \ref{Lemma:Substage:2:1}, we get inequalities as described in Lemma \ref{Lemma:Substage:2:1}. By Lemma \ref{Lemma:Substage:2:4}, the inequalities are ready to apply the Ackermann rules, and therefore we can eliminate all the propositional variables and $\mathsf{ALBA}^{\downarrow}_{\mathsf{Restricted}}$ succeeds on the input.
\end{proof}

\section{Completeness results}\label{Sec:Completeness}
In this section, we will prove that given any skeletal Sahlqvist formula $\phi\to\psi$, the logic $\mathbf{K}_{\mathcal{H}(@,\downarrow)}+(\phi\to\psi)$ is sound and strongly complete with respect to the class of frames defined by $\phi\to\psi$. Our proof strategy is as follows:

\begin{itemize}
\item First of all, we give a translation of each quasi-inequality $\&\mathsf{Reduce}(\phi_i\leq\psi_i)\Rightarrow \nomi_0\leq \neg\nomi_1$ in $\mathsf{Pure}(\phi\to\psi)$ into the language $\mathcal{L}(@,\downarrow)$ which results in a set of $\mathcal{L}(@,\downarrow)$-formulas $\{\pi_i\mid i\in I\}$, and we will show that $\phi\to\psi$ and the set $\Pi:=\{\pi_i\mid i\in I\}$ define the same class of Kripke frames.
\item Secondly, we prove that each $\pi_i$ is provable in $\mathbf{K}_{\mathcal{H}(@,\downarrow)}+(\phi\to\psi)$. Therefore, since $\mathbf{K}_{\mathcal{H}(@,\downarrow)}+\Pi$ is sound and strongly complete with respect to the class of frames defined by $\Pi$ (i.e.\ by $\phi\to\psi$), we get the soundness and strong completeness of $\mathbf{K}_{\mathcal{H}(@,\downarrow)}+(\phi\to\psi)$.
\end{itemize}

\subsection{The translation of inequalities and quasi-inequalities into $\mathcal{L}(@,\downarrow)$-formulas}
The key observation in the success proof in Section \ref{Sec:Success} is that in the systems in Stage 2 and the quasi-inequalities in Stage 3, for each inequality involved, either the left-hand side is $\nomi$ or $x$, or the right-hand side is $\neg\nomi$ or $\neg x$. Indeed, the inequality $\nomi\leq\gamma$ (resp.\ $x\leq\gamma$) is equivalent to the $\mathcal{L}(@,\downarrow)$-formula $@_\nomi \gamma$ (resp.\ $@_x \gamma$), and the inequality $\gamma\leq\neg\nomi$ (resp.\ $\gamma\leq\neg x$) is equivalent to the $\mathcal{L}(@,\downarrow)$-formula $\neg @_\nomi \gamma$ (resp.\ $\neg @_x \gamma$). Therefore, the systems obtained in Stage 2 and the quasi-inequalities obtained in Stage 3 are equivalent to a $\mathcal{L}(@,\downarrow)$-formula.

\begin{definition}[Translation of inequalities and quasi-inequalities into $\mathcal{L}(@,\downarrow)$-formulas]
We define the translation of the inequalities of the form $\nomi\leq\gamma$, $x\leq\gamma$, $\gamma\leq\neg\nomi$, $\gamma\leq\neg x$ into $\mathcal{L}(@,\downarrow)$-formulas as follows:
\begin{itemize}
\item $\mathsf{Tr}(\nomi\leq\gamma):=@_\nomi \gamma$;
\item $\mathsf{Tr}(x\leq\gamma):=@_x \gamma$;
\item $\mathsf{Tr}(\gamma\leq\neg\nomi):=\neg @_\nomi \gamma$;
\item $\mathsf{Tr}(\gamma\leq\neg x):=\neg @_x \gamma$.
\end{itemize}
When an inequality is of more than one of the forms above at the same time, we can take any form appearing in the list since they are equivalent.

Given a quasi-inequality $\mathsf{Quasi}$ of the form $\mathsf{Ineq}_1\ \&\ \ldots\ \&\ \mathsf{Ineq}_n\ \Rightarrow\ \nomi\leq\neg\nomj$ where each of $\mathsf{Ineq}_1, \ldots, \mathsf{Ineq}_n$ is of the form $\nomi\leq\gamma$, $x\leq\gamma$, $\gamma\leq\neg\nomi$ or $\gamma\leq\neg x$, we define  
$$\mathsf{Tr}(\mathsf{Quasi}):=\mathsf{Tr}(\mathsf{Ineq}_1)\land\ldots\land\mathsf{Tr}(\mathsf{Ineq}_n)\to\neg @_{\nomi}\nomj$$

Given a set $\mathsf{QuasiSet}$ of quasi-inequalities of the form above, we define $$\mathsf{Tr}(\mathsf{QuasiSet}):=\bigwedge_{\mathsf{Quasi}\in\mathsf{QuasiSet}}\mathsf{Tr}(\mathsf{Quasi}).$$
\end{definition}

\begin{proposition}\label{Prop:Translation}
For each inequality $\mathsf{Ineq}$ of the form $\nomi\leq\gamma$, $x\leq\gamma$, $\gamma\leq\neg\nomi$ or $\gamma\leq\neg x$, each quasi-inequality $\mathsf{Quasi}$ of the form described in the definition above, we have that for any Kripke model $\mathbb{M}$, any assignment $g$ on $\mathbb{M}$, 
$$\mathbb{M},g\Vdash\mathsf{Ineq}\mbox{ iff }\mathbb{M},g\Vdash\mathsf{Tr}(\mathsf{Ineq})$$
$$\mathbb{M},g\Vdash\mathsf{Quasi}\mbox{ iff }\mathbb{M},g\Vdash\mathsf{Tr}(\mathsf{Quasi}).$$
$$\mathbb{M},g\Vdash\mathsf{QuasiSet}\mbox{ iff }\mathbb{M},g\Vdash\mathsf{Tr}(\mathsf{QuasiSet}).$$
\end{proposition}

\subsection{Provability of the translations}

\begin{lemma}
Given an input skeletal Sahlqvist formula $\phi\to\psi$, during Stage 1, for each inequality $\theta\leq\chi$ produced by the algorithm, we have $\vdash_{\phi\to\psi}\theta\to\chi$.
\end{lemma}

\begin{proof}
We prove by induction on the algorithm steps in Stage 1 that for each inequality $\theta\leq\chi$ produced by the algorithm, we have $\vdash_{\phi\to\psi}\theta\to\chi$.

\begin{itemize}
\item For the basic step, obviously $\vdash_{\phi\to\psi}\phi\to\psi$. 
\item For the distribution rules, we have that the following equivalences are provable in $\mathbf{K}_{\mathcal{H}(@,\downarrow)}$ (therefore in $\mathbf{K}_{\mathcal{H}(@,\downarrow)}+(\phi\to\psi)$), thus by the (Res) rule, for the inequality $\theta\leq\chi$ obtained by the distribution rule, we have $\vdash_{\phi\to\psi}\theta\to\chi$.

\begin{itemize}
\item $\Diamond(\alpha\lor\beta)\leftrightarrow\Diamond\alpha\lor\Diamond\beta$;
\item $\neg(\alpha\lor\beta)\leftrightarrow\neg\alpha\land\neg\beta$;
\item $(\alpha\lor\beta)\land\gamma\leftrightarrow(\alpha\land\gamma)\lor(\beta\land\gamma)$;
\item $\alpha\land(\beta\lor\gamma)\leftrightarrow(\alpha\land\beta)\lor(\alpha\land\gamma)$;
\item $\downarrow x.(\alpha\lor\beta)\leftrightarrow (\downarrow x.\alpha\lor\downarrow x.\beta)$;
\item $@_{\nomi}(\alpha\lor\beta)\leftrightarrow (@_{\nomi}\alpha\lor @_{\nomi}\beta)$;
\item $@_{x}(\alpha\lor\beta)\leftrightarrow (@_{x}\alpha\lor @_{x}\beta)$;
\item $((\alpha\lor\beta)\to\gamma)\leftrightarrow((\alpha\to\gamma)\land(\beta\to\gamma))$;
\item $\Box(\alpha\land\beta)\leftrightarrow\Box\alpha\land\Box\beta$;
\item $\neg(\alpha\land\beta)\leftrightarrow\neg\alpha\lor\neg\beta$;
\item $(\alpha\land\beta)\lor\gamma\leftrightarrow(\alpha\lor\gamma)\land(\beta\lor\gamma)$;
\item $\alpha\lor(\beta\land\gamma)\leftrightarrow(\alpha\lor\beta)\land(\alpha\lor\gamma)$;
\item $\downarrow x.(\alpha\land\beta)\leftrightarrow (\downarrow x.\alpha\land \downarrow x.\beta)$;
\item $@_{\nomi}(\alpha\land\beta)\leftrightarrow (@_{\nomi}\alpha\land @_{\nomi}\beta)$;
\item $@_{x}(\alpha\land\beta)\leftrightarrow (@_{x}\alpha\land @_{x}\beta)$;
\item $(\alpha\to\beta\land\gamma)\leftrightarrow(\alpha\to\beta)\land(\alpha\to\gamma)$.
\end{itemize}

\item For the splitting rules, suppose we have $\alpha\lor\beta\leq\gamma$. By induction hypothesis, we have $\vdash_{\phi\to\psi}\alpha\lor\beta\to\gamma$. By classical propositional logic, we have that $\vdash_{\phi\to\psi}\alpha\to\gamma$ and $\vdash_{\phi\to\psi}\beta\to\gamma$. Similarly we have that from $\vdash_{\phi\to\psi}\alpha\to\beta\land\gamma$ we can get $\vdash_{\phi\to\psi}\alpha\to\beta$ and $\vdash_{\phi\to\psi}\alpha\to\gamma$.

\item For the monotone and antitone variable-elimination rules, suppose we have $\alpha(p)\leq\beta(p)$. By induction hypothesis, we have $\vdash_{\phi\to\psi}\alpha(p)\to\beta(p)$. By uniform substitution, we have $\vdash_{\phi\to\psi}\alpha(\bot)\to\beta(\bot)$. Similarly, from $\vdash_{\phi\to\psi}\beta(p)\to\alpha(p)$ we can get $\vdash_{\phi\to\psi}\beta(\top)\to\alpha(\top)$.
\end{itemize}
\end{proof}

\begin{corollary}\label{Cor:Stage:1:Before:Fir:Approx}
Suppose that for the input formula $\phi\to\psi$, in Stage 1, before the first-approximation rule, we get a set of inequalities $\{\theta_1\leq\chi_1, \ldots, \theta_n\leq\chi_n\}$, then $\vdash_{\phi\to\psi}\theta_i\to\chi_i$ for $1\leq i\leq n$.
\end{corollary}

\begin{lemma}
For each $\theta_i\leq\chi_i$ before the first-approximation rule in Stage 1, after the first-approximation rule, we get the system $\mathsf{Sys}_i:=\{\nomi_0\leq\theta_i, \chi_i\leq\neg\nomi_1\}$, which corresponds to the quasi-inequality $\mathsf{Quasi}_i:=\nomi_0\leq\theta_i\ \&\ \chi_i\leq\neg\nomi_1\ \Rightarrow\ \nomi_0\leq\neg\nomi_1$, then we have that $\vdash_{\phi\to\psi}\mathsf{Tr}(\mathsf{Quasi}_i)$.
\end{lemma}

\begin{proof}
By Corollary \ref{Cor:Stage:1:Before:Fir:Approx}, we have that $\vdash_{\phi\to\psi}\theta_i\to\chi_i$. Therefore we have the following proof in $\mathbf{K}_{\mathcal{H}(@,\downarrow)}+(\phi\to\psi)$\footnote{In the proof, CPC means using classical propositional logic.}:
\begin{center}
\begin{tabular}{l l l}
1 & $\vdash_{\phi\to\psi}\theta_i\to\chi_i$ & (Assumption)\\

2 & $\vdash_{\phi\to\psi}@_{\nomi_0}(\theta_i\to\chi_i)$ & (Nec$_{@}$, 1)\\

3 & $\vdash_{\phi\to\psi}@_{\nomi_0}(\theta_i\to\chi_i)\to(@_{\nomi_0}\theta_i\to @_{\nomi_0}\chi_i)$ & (K$_{@}$)\\

4 & $\vdash_{\phi\to\psi}@_{\nomi_0}\theta_i\to @_{\nomi_0}\chi_i$ & (MP, 2, 3)\\

5 & $\vdash_{\phi\to\psi}@_{\nomi_0}\theta_i\land @_{\nomi_1}\nomi_0\to @_{\nomi_1}\theta_i$ & (Trans)\\

6 & $\vdash_{\phi\to\psi}@_{\nomi_0}\theta_i\land @_{\nomi_1}\nomi_0\to @_{\nomi_1}\chi_i$ & (CPC, 4, 5)\\

7 & $\vdash_{\phi\to\psi}@_{\nomi_0}\nomi_1\to@_{\nomi_1}\nomi_0$ & (Sym)\\

8 & $\vdash_{\phi\to\psi}@_{\nomi_0}\theta_i\land @_{\nomi_0}\nomi_1\to @_{\nomi_1}\chi_i$ & (CPC, 6, 7)\\

9 & $\vdash_{\phi\to\psi}@_{\nomi_0}\theta_i\land\neg @_{\nomi_1}\chi_i\to\neg @_{\nomi_0}\nomi_1$ & (CPC, 8)\\
\end{tabular}
\end{center}
Therefore $\vdash_{\phi\to\psi}\mathsf{Tr}(\mathsf{Quasi}_i)$.
\end{proof}

Now we fix a quasi-inequality $\mathsf{Quasi}_i:=\nomi_0\leq\theta_i\ \&\ \chi_i\leq\neg\nomi_1\ \Rightarrow\ \nomi_0\leq\neg\nomi_1$. We will prove that for each system $\mathsf{Sys}$ obtained during the Stage 2, $\mathsf{Tr}(\&\mathsf{Sys}\to\nomi_0\leq\neg\nomi_1)$ is provable.

\begin{lemma}
Given a quasi-inequality $\mathsf{Quasi}_i:=\nomi_0\leq\theta_i\ \&\ \chi_i\leq\neg\nomi_1\ \Rightarrow\ \nomi_0\leq\neg\nomi_1$ obtained in Stage 1, for each system $\mathsf{Sys}$ obtained during the Stage 2, $\vdash_{\phi\to\psi}\mathsf{Tr}(\&\mathsf{Sys}\to\nomi_0\leq\neg\nomi_1)$.
\end{lemma}

\begin{proof}
First of all, since in each inequality in the system, either the left-hand side is a nominal/state variable, or the right-hand side is the negation of a nominal/state variable, $\&\mathsf{Sys}\to\nomi_0\leq\neg\nomi_1$ can be translated.

We prove by induction on the algorithm steps in Stage 2 that for each system $\mathsf{Sys}$ obtained during the Stage 2, $\vdash_{\phi\to\psi}\mathsf{Tr}(\&\mathsf{Sys}\to\nomi_0\leq\neg\nomi_1)$ is provable.

\begin{itemize}
\item For the basic step, obviously $\vdash_{\phi\to\psi}\mathsf{Tr}(\mathsf{Quasi}_i)$.
\item For the splitting rules, it suffices to prove that from $\vdash_{\phi\to\psi}@_\nomi(\beta\land\gamma)\land\alpha\to\neg@_{\nomi_0}\nomi_1$ one can get $\vdash_{\phi\to\psi}@_\nomi\beta\land@_\nomi\gamma\land\alpha\to\neg@_{\nomi_0}\nomi_1$ and from $\vdash_{\phi\to\psi}\neg@_\nomi(\alpha\lor\beta)\land\gamma\to\neg@_{\nomi_0}\nomi_1$ one can get $\vdash_{\phi\to\psi}\neg@_\nomi\alpha\land\neg@_\nomi\beta\land\gamma\to\neg@_{\nomi_0}\nomi_1$, which follows by the facts that $\vdash_{\phi\to\psi}@_\nomi(\beta\land\gamma)\leftrightarrow@_\nomi\beta\land@_\nomi\gamma$ and $\vdash_{\phi\to\psi}\neg@_\nomi(\alpha\lor\beta)\leftrightarrow\neg@_\nomi\alpha\land\neg@_\nomi\beta$.

\item For the approximation rule from $\nomi\leq\Diamond\alpha$ to $\nomj\leq\alpha$ and $\nomi\leq\Diamond\nomj$, it suffices to prove that from $\vdash_{\phi\to\psi}@_{\nomi}\Diamond\alpha\land\gamma\to\neg@_{\nomi_0}\nomi_1$ one can get $\vdash_{\phi\to\psi}@_{\nomj}\alpha\land@_{\nomi}\Diamond\nomj\land\gamma\to\neg@_{\nomi_0}\nomi_1$, which follows by the fact that $\vdash_{\phi\to\psi}@_{\nomj}\alpha\land@_{\nomi}\Diamond\nomj\to@_{\nomi}\Diamond\alpha$. 

For the approximation rule from $x\leq\Diamond\alpha$ to $\nomj\leq\alpha$ and $x\leq\Diamond\nomj$, the proof is similar. 
\item For the approximation rule from $\Box\alpha\leq\neg\nomi$ to $\alpha\leq\neg\nomj$ and $\Box\neg\nomj\leq\neg\nomi$, it suffices to prove that from $\vdash_{\phi\to\psi}\neg@_{\nomi}\Box\alpha\land\gamma\to\neg@_{\nomi_0}\nomi_1$ one can get $\vdash_{\phi\to\psi}\neg@_{\nomi}\Box\neg\nomj\land\neg @_{\nomj}\alpha\land\gamma\to\neg@_{\nomi_0}\nomi_1$, which follows from the fact that $\vdash_{\phi\to\psi}\neg@_{\nomi}\Box\neg\nomj\land\neg @_{\nomj}\alpha\to\neg@_{\nomi}\Box\alpha$ (i.e.\ $\vdash_{\phi\to\psi}@_{\nomi}\Diamond\nomj\land @_{\nomj}\neg\alpha\to @_{\nomi}\Diamond\neg\alpha$).

For the approximation rule from $\Box\alpha\leq\neg x$ to $\alpha\leq\neg\nomj$ and $\Box\neg\nomj\leq\neg x$, the proof is similar.
\item For the approximation rule from $\nomi\leq @_{\nomj}\alpha$ to $\nomj\leq\alpha$, it suffices to prove that from $\vdash_{\phi\to\psi}@_{\nomi}@_{\nomj}\alpha\land\gamma\to\neg@_{\nomi_0}\nomi_1$ one can get $\vdash_{\phi\to\psi}@_{\nomj}\alpha\land\gamma\to\neg@_{\nomi_0}\nomi_1$, which follows from the fact that $\vdash_{\phi\to\psi}@_{\nomj}\alpha\leftrightarrow@_{\nomi}@_{\nomj}\alpha$.

For the approximation rule from $x\leq @_{\nomj}\alpha$ to $\nomj\leq\alpha$, from $\nomi\leq @_x \alpha$ to $x\leq\alpha$, from $y\leq @_x \alpha$ to $x\leq\alpha$, the proof is similar.
\item For the approximation rule from $@_{\nomj}\alpha\leq\neg\nomi$ to $\alpha\leq\neg\nomj$, it suffices to prove that from $\vdash_{\phi\to\psi}\neg @_{\nomi}@_{\nomj}\alpha\land\gamma\to\neg@_{\nomi_0}\nomi_1$ one can get $\vdash_{\phi\to\psi}\neg @_{\nomj}\alpha\land\gamma\to\neg@_{\nomi_0}\nomi_1$, which follows from the fact that $\vdash_{\phi\to\psi}@_{\nomj}\alpha\leftrightarrow @_{\nomi}@_{\nomj}\alpha$. 

For the approximation rule from $@_{\nomj}\alpha\leq\neg x$ to $\alpha\leq\neg\nomj$, from $@_{x}\alpha\leq\neg \nomi$ to $\alpha\leq\neg x$, from $@_{x}\alpha\leq\neg y$ to $\alpha\leq\neg x$, the proof is similar.

\item For the approximation rule from $\nomi\leq\downarrow x.\alpha$ to $\nomi\leq\alpha[\nomi/x]$, it suffices to prove that from $\vdash_{\phi\to\psi}@_{\nomi}\downarrow x.\alpha\land\gamma\to\neg@_{\nomi_0}\nomi_1$ one can get $\vdash_{\phi\to\psi}@_{\nomi}\alpha[\nomi/x]\land\gamma\to\neg@_{\nomi_0}\nomi_1$, which follows from the fact that $\vdash_{\phi\to\psi}@_{\nomi}\alpha[\nomi/x]\leftrightarrow @_{\nomi}\downarrow x.\alpha$. 

For the approximation rule from $y\leq\downarrow x.\alpha$ to $y\leq\alpha[y/x]$, the proof is similar.

\item For the approximation rule from $\downarrow x.\alpha\leq\neg\nomi$ to $\alpha[\nomi/x]\leq\neg\nomi$, it suffices to prove that from $\vdash_{\phi\to\psi}\neg @_{\nomi}\downarrow x.\alpha\land\gamma\to\neg@_{\nomi_0}\nomi_1$ one can get $\vdash_{\phi\to\psi}\neg @_{\nomi}\alpha[\nomi/x]\land\gamma\to\neg@_{\nomi_0}\nomi_1$, which follows from the fact that $\vdash_{\phi\to\psi}@_{\nomi}\alpha[\nomi/x]\leftrightarrow @_{\nomi}\downarrow x.\alpha$. 

For the approximation rule from $\downarrow x.\alpha\leq\neg y$ to $\alpha[y/x]\leq\neg y$, the proof is similar.

\item For the approximation rule from $\alpha\to\beta\leq\neg\nomi$ to $\nomj\leq\alpha$, $\beta\leq\neg\nomk$ and $\nomj\to\neg\nomk\leq\neg\nomi$, it suffices to prove that from $\vdash_{\phi\to\psi}\neg @_{\nomi}(\alpha\to\beta)\land\gamma\to\neg@_{\nomi_0}\nomi_1$ one can get $\vdash_{\phi\to\psi}@_{\nomj}\alpha\land\neg@_{\nomk}\beta\land\neg@_{\nomi}(\nomj\to\neg\nomk)\land\gamma\to\neg@_{\nomi_0}\nomi_1$, which follows from the fact that $\vdash_{\phi\to\psi}@_{\nomj}\alpha\land\neg@_{\nomk}\beta\land\neg@_{\nomi}(\nomj\to\neg\nomk)\to\neg @_{\nomi}(\alpha\to\beta)$.

For the approximation rule from $\alpha\to\beta\leq\neg x$ to $\nomj\leq\alpha$, $\beta\leq\neg\nomk$ and $\nomj\to\neg\nomk\leq\neg x$, the proof is similar. 
\item For the residuation rule from $\nomi\leq\neg\alpha$ to $\alpha\leq\neg\nomi$, it suffices to prove that from $\vdash_{\phi\to\psi}@_{\nomi}\neg\alpha\land\gamma\to\neg@_{\nomi_0}\nomi_1$ one can get $\vdash_{\phi\to\psi}\neg @_{\nomi}\alpha\land\gamma\to\neg@_{\nomi_0}\nomi_1$, which follows from the fact that $\vdash_{\phi\to\psi}\neg @_{\nomi}\alpha\leftrightarrow@_{\nomi}\neg\alpha$.

For the residuation rule from $x\leq\neg\alpha$ to $\alpha\leq\neg x$, the proof is similar.
\item For the residuation rule from $\neg\alpha\leq\neg\nomi$ to $\nomi\leq\alpha$, it suffices to prove that from $\vdash_{\phi\to\psi}\neg @_{\nomi}\neg\alpha\land\gamma\to\neg@_{\nomi_0}\nomi_1$ one can get $\vdash_{\phi\to\psi}@_{\nomi}\alpha\land\gamma\to\neg@_{\nomi_0}\nomi_1$, which follows from the fact that $\vdash_{\phi\to\psi}\neg @_{\nomi}\neg\alpha\leftrightarrow@_{\nomi}\alpha$.

For the residuation rule from $\neg\alpha\leq\neg x$ to $x\leq\alpha$, the proof is similar.
\item For the right-handed Ackermann rule from  
$\left\{ \begin{array}{ll}
\nomi_1\leq p \\
\vdots\\
\nomi_n\leq p \\
\nomj_1\leq\gamma_1\\
\vdots\\
\nomj_m\leq\gamma_m\\
\beta_1\leq\neg\nomk_1\\
\vdots\\
\beta_k\leq\neg\nomk_k\\
\end{array} \right.$ 
to
$\left\{ \begin{array}{ll}
\nomj_1\leq\gamma_1[(\nomi_1\lor\ldots\lor\nomi_n)/p]\\
\vdots\\
\nomj_m\leq\gamma_m[(\nomi_1\lor\ldots\lor\nomi_n)/p]\\
\beta_1[(\nomi_1\lor\ldots\lor\nomi_n)/p]\leq\neg\nomk_1\\
\vdots\\
\beta_k[(\nomi_1\lor\ldots\lor\nomi_n)/p]\leq\neg\nomk_k\\
\end{array} \right.$

without loss of generality we suppose that $m=k=1$, then it suffices to prove that from 
$$\vdash_{\phi\to\psi}@_{\nomi_1}p\land\ldots\land@_{\nomi_n}p\land@_{\nomj_1}\gamma_1\land\neg@_{\nomk_1}\beta_1\to\neg@_{\nomi_0}\nomi_1$$ one can get $$\vdash_{\phi\to\psi}@_{\nomj_1}\gamma_1[(\nomi_1\lor\ldots\lor\nomi_n)/p]\land\neg@_{\nomk_1}\beta_1[(\nomi_1\lor\ldots\lor\nomi_n)/p]\to\neg@_{\nomi_0}\nomi_1.$$ By uniform substitution $(\nomi_1\lor\ldots\lor\nomi_n)/p$ on $$\vdash_{\phi\to\psi}@_{\nomi_1}p\land\ldots\land@_{\nomi_n}p\land@_{\nomj_1}\gamma_1\land\neg@_{\nomk_1}\beta_1\to\neg@_{\nomi_0}\nomi_1,$$ we can get $$\vdash_{\phi\to\psi}@_{\nomi_1}(\nomi_1\lor\ldots\lor\nomi_n)\land\ldots\land@_{\nomi_n}(\nomi_1\lor\ldots\lor\nomi_n)\land@_{\nomj_1}\gamma_1[(\nomi_1\lor\ldots\lor\nomi_n)/p]\land\neg@_{\nomk_1}\beta_1[(\nomi_1\lor\ldots\lor\nomi_n)/p]\to\neg@_{\nomi_0}\nomi_1.$$ Since $\vdash_{\phi\to\psi}@_{\nomi_1}(\nomi_1\lor\ldots\lor\nomi_n)$, \ldots, $\vdash_{\phi\to\psi}@_{\nomi_n}(\nomi_1\lor\ldots\lor\nomi_n)$, we have that $$\vdash_{\phi\to\psi}@_{\nomj_1}\gamma_1[(\nomi_1\lor\ldots\lor\nomi_n)/p]\land\neg@_{\nomk_1}\beta_1[(\nomi_1\lor\ldots\lor\nomi_n)/p]\to\neg@_{\nomi_0}\nomi_1.$$
\item For the left-handed Ackermann rule from  
$\left\{ \begin{array}{ll}
p\leq\neg\nomi_1 \\
\vdots\\
p\leq\neg\nomi_n \\
\nomj_1\leq\gamma_1\\
\vdots\\
\nomj_m\leq\gamma_m\\
\beta_1\leq\neg\nomk_1\\
\vdots\\
\beta_k\leq\neg\nomk_k\\
\end{array} \right.$
to
$\left\{ \begin{array}{ll}
\nomj_1\leq\gamma_1[(\neg\nomi_1\land\ldots\land\neg\nomi_n)/p]\\
\vdots\\
\nomj_m\leq\gamma_m[(\neg\nomi_1\land\ldots\land\neg\nomi_n)/p]\\
\beta_1[(\neg\nomi_1\land\ldots\land\neg\nomi_n)/p]\leq\neg\nomk_1\\
\vdots\\
\beta_m[(\neg\nomi_1\land\ldots\land\neg\nomi_n)/p]\leq\neg\nomk_k\\
\end{array} \right.$

without loss of generality we suppose that $m=k=1$, then it suffices to prove that from $$\vdash_{\phi\to\psi}\neg@_{\nomi_1}p\land\ldots\land\neg@_{\nomi_n}p\land@_{\nomj_1}\gamma_1\land\neg@_{\nomk_1}\beta_1\to\neg@_{\nomi_0}\nomi_1$$ one can get $$\vdash_{\phi\to\psi}@_{\nomj_1}\gamma_1[(\neg\nomi_1\land\ldots\land\neg\nomi_n)/p]\land\neg@_{\nomk_1}\beta_1[(\neg\nomi_1\land\ldots\land\neg\nomi_n)/p]\to\neg@_{\nomi_0}\nomi_1.$$ By uniform substitution $(\neg\nomi_1\land\ldots\land\neg\nomi_n)/p$ on $$\vdash_{\phi\to\psi}\neg@_{\nomi_1}p\land\ldots\land\neg@_{\nomi_n}p\land@_{\nomj_1}\gamma_1\land\neg@_{\nomk_1}\beta_1\to\neg@_{\nomi_0}\nomi_1,$$ we can get $$\vdash_{\phi\to\psi}\neg@_{\nomi_1}(\neg\nomi_1\land\ldots\land\neg\nomi_n)\land\ldots\land\neg@_{\nomi_n}(\neg\nomi_1\land\ldots\land\neg\nomi_n)\land@_{\nomj_1}\gamma_1[(\neg\nomi_1\land\ldots\land\neg\nomi_n)/p]\land\neg@_{\nomk_1}\beta_1[(\neg\nomi_1\land\ldots\land\neg\nomi_n)/p]\to\neg@_{\nomi_0}\nomi_1.$$ Since $\vdash_{\phi\to\psi}\neg@_{\nomi_1}(\neg\nomi_1\land\ldots\land\neg\nomi_n)$, \ldots, $\vdash_{\phi\to\psi}\neg@_{\nomi_n}(\neg\nomi_1\land\ldots\land\neg\nomi_n)$, we have that $$\vdash_{\phi\to\psi}@_{\nomj_1}\gamma_1[(\neg\nomi_1\land\ldots\land\neg\nomi_n)/p]\land\neg@_{\nomk_1}\beta_1[(\neg\nomi_1\land\ldots\land\neg\nomi_n)/p]\to\neg@_{\nomi_0}\nomi_1.$$
\end{itemize}
\end{proof}

\begin{corollary}\label{Cor:Main}
Given a skeletal Sahlqvist formula $\phi\to\psi$, for each quasi-inequality $\mathsf{Quasi}$ in $\mathsf{Pure}(\phi\to\psi)$, we have that $\vdash_{\phi\to\psi}\mathsf{Tr}(\mathsf{Quasi})$, therefore $\vdash_{\phi\to\psi}\mathsf{Tr}(\mathsf{Pure}(\phi\to\psi))$.
\end{corollary}

\begin{proof}
It suffices to see that for each pure quasi-inequality produced after Stage 2, by uniformly substitute free occurrences of state variables by fresh nominals, the translation of the resulting pure quasi-inequality is still provable in $\mathbf{K}_{\mathcal{H}(@,\downarrow)}+(\phi\to\psi)$.
\end{proof}

\subsection{Main Proof}

Now we are ready to prove our main result:

\begin{theorem}
For any skeletal Sahlqvist formula $\phi\to\psi$, $\mathbf{K}_{\mathcal{H}(@,\downarrow)}+(\phi\to\psi)$ is sound and strongly complete with respect to the class of Kripke frames $\mathcal{F}$ defined by $\phi\to\psi$.
\end{theorem}

\begin{proof}
Our proof strategy is as follows: we prove that for any $\mathcal{L}(@,\downarrow)$-formula set $\Gamma$ and any $\mathcal{L}(@,\downarrow)$-formula $\gamma$, 
$$\Gamma\vdash_{\phi\to\psi}\gamma\ \Rightarrow\ \Gamma\Vdash_{\mathcal{F}}\gamma\ \Rightarrow\ \Gamma\vdash_{\mathsf{Tr}(\mathsf{Pure}(\phi\to\psi))}\gamma\ \Rightarrow\ \Gamma\vdash_{\phi\to\psi}\gamma.$$

\begin{itemize}
\item For the first implication, i.e.\ the soundness part, it is easy. 
\item For the second implication, from the fact that 
\begin{center}
\begin{tabular}{l l l}
& $\mathbb{F}\Vdash\phi\to\psi$ &\\
iff & $\mathbb{F}\Vdash\mathsf{Pure}(\phi\to\psi)$ & (Theorem \ref{Thm:Soundness})\\
iff & $\mathbb{F}\Vdash\mathsf{Tr}(\mathsf{Pure}(\phi\to\psi))$ & (corollary of Proposition \ref{Prop:Translation})\\
\end{tabular}
\end{center}
we have that $\mathcal{F}$ is also defined by $\mathsf{Tr}(\mathsf{Pure}(\phi\to\psi))$. By Theorem \ref{Completeness:Pure}, we have the completeness of $\mathbf{K}_{\mathcal{H}(@,\downarrow)}+\mathsf{Tr}(\mathsf{Pure}(\phi\to\psi))$ with respect to $\mathcal{F}$.
\item For the third implication, it suffices to show that all theorems of $\mathbf{K}_{\mathcal{H}(@,\downarrow)}+\mathsf{Tr}(\mathsf{Pure}(\phi\to\psi))$ are also theorems of $\mathbf{K}_{\mathcal{H}(@,\downarrow)}+(\phi\to\psi)$. To show this, it is enough to prove that $\vdash_{\phi\to\psi}\mathsf{Tr}(\mathsf{Pure}(\phi\to\psi))$, which follows from Corollary \ref{Cor:Main}.
\end{itemize}
\end{proof}
By an easy adaptation of the previous results to a set $\Sigma$ of skeletal Sahlqvist formulas, we have the following corollary:
\begin{corollary}
For any set $\Sigma$ of skeletal Sahlqvist formulas, $\mathbf{K}_{\mathcal{H}(@,\downarrow)}+\Sigma$ is sound and strongly complete with respect to the class of Kripke frames $\mathcal{F}$ defined by $\Sigma$.
\end{corollary}

\section{Conclusion}\label{Sec:Conclusion}

In the present paper, we investigates the completeness theory for hybrid logic with binder $\mathcal{L}(@,\downarrow)$. We define the class of skeletal Sahlqvist formulas, and show that for any set $\Sigma$ of skeletal Sahlqvist formulas, $\mathbf{K}_{\mathcal{H}(@,\downarrow)}+\Sigma$ is sound and strongly complete with respect to the class of Kripke frames $\mathcal{F}$ defined by $\Sigma$. Our strategy is to use the algorithm $\mathsf{ALBA}^{\downarrow}_{\mathsf{Restricted}}$ to transform an input skeletal Sahlqvist formula $\phi\to\psi$ into an equivalent $\mathcal{L}(@,\downarrow)$-formula $\mathsf{Tr}(\mathsf{Pure}(\phi\to\psi))$, and then show that $\mathbf{K}_{\mathcal{H}(@,\downarrow)}+(\phi\to\psi)$ proves $\mathsf{Tr}(\mathsf{Pure}(\phi\to\psi))$.

Our methodology could also work for $\mathcal{L}(@)$, which follows from a restricted version of the algorithm $\mathsf{hybrid}$-$\mathsf{ALBA}$ defined in \cite{ConRob}. Indeed, we got inspiration of the definition of skeletal Sahlqivst inequalities from \cite{ConRob}. In \cite{ConRob}, Conradie and Robinson gave an algebraic proof of the completeness of $\mathbf{K}_{\mathcal{H}(@)}+\Sigma$ where $\Sigma$ is a set of skeletal formulas. Our proof can be seen as a proof-theoretic counterpart of their proof.

For future directions, we list the following:

\begin{itemize}
\item In \cite{tCMaVi06}, ten Cate, Marx and Viana proved that modal Sahlqvist formulas that do not contain occurrences of nominals axiomatize complete logics extending $\mathbf{K}_{\mathcal{H}(@)}$. A future question is whether this result could be extended to the language $\mathcal{L}(@,\downarrow)$.
\item In \cite{ConRob}, Conradie and Robinson proved that for any set $\Sigma$ of nominally skeletal inductive formulas, the logic $\mathbf{K}_{\mathcal{H}(@)}+\Sigma$ is sound and strongly complete with respect to its class of Kripke frames. A future question is that whether this result could be extended to the language $\mathcal{L}(@,\downarrow)$.
\item In \cite{Li06}, Litak gave an algebraization of hybrid logic with binder $\mathcal{H}(@,\downarrow)$. A future question is whether we can use this algebraization to give canonicity proofs of certain formulas to prove completeness results.
\end{itemize}

\paragraph{Acknowledgement} The research of the author is supported by the Taishan Young Scholars Program of the Government of Shandong Province, China (No.tsqn201909151).

\bibliographystyle{abbrv}
\bibliography{Binder}

\end{document}